\documentclass[a4paper,11pt,fleqn]{article}

\usepackage[T1]{fontenc} 
\usepackage[utf8]{inputenc} 
\usepackage{lmodern} 
\usepackage{microtype} 
\usepackage[margin=2.5cm]{geometry}
\usepackage{ellipsis} 
\usepackage{cite} 

\usepackage{authblk}
\usepackage{amsmath} 
\usepackage{amssymb}
\usepackage{dsfont}
\usepackage[colorlinks,allcolors=blue]{hyperref}
\usepackage[nameinlink]{cleveref}
\usepackage{tikz}
\usepackage[font=footnotesize,labelsep=period,labelfont=bf,margin=3em]{caption}
\usepackage{amsmath,amssymb}
\usepackage[amsthm,thmmarks]{ntheorem}

\newtheorem{theorem}{Theorem}
\newtheorem{definition}{Definition}

\newtheorem{lemma}{Lemma}

\newtheorem{remark}{Remark}

\crefname{chapter}{Chapter}{Chapters}
\crefname{section}{Section}{Sections}
\crefname{subsection}{Subsection}{Subsections}
\crefname{definition}{Definition}{Definitions}
\crefname{example}{Example}{Examples}
\crefname{figure}{Figure}{Figures}
\crefname{table}{Table}{Tables}
\crefname{theorem}{Theorem}{Theorems}
\crefname{lemma}{Lemma}{Lemmata}
\crefname{remark}{Remark}{Remarks}
\crefname{corollary}{Corollary}{Corollary}
\crefname{observation}{Observation}{Observations}
\crefname{fact}{Fact}{Facts}
\crefname{conjecture}{Conjecture}{Conjectures}
\crefname{equation}{}{}
\crefname{enumi}{}{}

\title{\textbf{Forbidden Subgraph Problems with Predictions}}

\author{Hans-Joachim Böckenhauer}
\author{Melvin Jahn}
\author{Dennis Komm}
\author{Moritz Stocker}

\affil{Department of Computer Science, ETH Zurich\\%
  \small\texttt{\{hjb, dennis.komm, moritz.stocker\}@inf.ethz.ch, mejahn@student.ethz.ch}}

\date{\today}

\begin{document}

\maketitle

\makeatletter
\def\moverlay{\mathpalette\mov@rlay}
\def\mov@rlay#1#2{\leavevmode\vtop{%
   \baselineskip\z@skip \lineskiplimit-\maxdimen
   \ialign{\hfil$\m@th#1##$\hfil\cr#2\crcr}}}
\newcommand{\charfusion}[3][\mathord]{
    #1{\ifx#1\mathop\vphantom{#2}\fi
        \mathpalette\mov@rlay{#2\cr#3}
      }
    \ifx#1\mathop\expandafter\displaylimits\fi}
\makeatother

\newcommand{\cupdot}{\charfusion[\mathbin]{\cup}{\cdot}}
\newcommand{\bigcupdot}{\charfusion[\mathop]{\bigcup}{\boldsymbol\cdot}}

\begin{abstract}
\noindent In the \textsc{Online Delayed Connected $H$-Node-Deletion Problem},
an unweighted graph is revealed vertex by vertex and it must remain free
of any induced copies of a specific connected induced forbidden subgraph $H$ at
each point in time. To achieve this, an algorithm must, upon each occurrence of
$H$, identify and irrevocably delete one or more vertices. The objective is to
delete as few vertices as possible. We provide tight bounds on the competitive
ratio for forbidden subgraphs $H$ that do not contain two true twins or that do
not contain two false twins.

We further consider the problem within the model of predictions, where the
algorithm is provided with a single bit of advice for each revealed vertex.
These predictions are considered to be provided by an untrusted source and may
be incorrect. We present a family of algorithms solving the \textsc{Online
Delayed Connected $H$-Node-Deletion Problem} with predictions and show that
it is Pareto-optimal with respect to competitivity and robustness for the
online vertex cover problem for 2-connected forbidden subgraphs that do not
contain two true twins or that do not contain two false twins, as well as for
forbidden paths of length greater than four. We also propose subgraphs for
which a better algorithm might exist.
\end{abstract}

\section{Introduction} \label{chap:intro} 
Online algorithms receive their input piece by piece as a sequence of \emph{requests}
and have to make irrevocable decisions, called \emph{answers} upon each such request. As a result, they
face a significant disadvantage compared to traditional algorithms, as they
must act on portions of the instance without having complete information about
it. 

The performance of an online algorithm on an optimization problem is
classically measured by comparing its cost to the optimal offline solution,
which is the best possible solution given complete knowledge of the input in
advance. The \textit{competitive ratio} of an online algorithm is defined as
the worst-case ratio of its cost on an instance compared to the cost of the
optimal offline solution~\cite{borodin1998online}. 
Some models attempt to further analyze the online nature of a problem by asking
how much additional information (so-called \emph{advice})
is needed to improve the competitive ratio of an online algorithm.
\textit{Advice complexity} was first studied by
Dobrev et al.~\cite{dobrev2009measuring} and further refined
by Böckenhauer et al.\ \cite{bockenhauer2009advice}, Hromkovi\v{c} et al.~\cite{hromkovivc2010information},
and Emek et al.~\cite{emek2009online,emek2011online}.
A natural trade-off arises between the size of the advice and the performance of the
algorithm utilizing this information. Of course, if the advice is sufficiently large to encode the
entire optimal offline solution, the algorithm can typically obtain a
competitive ratio of $1$. There is a broad survey on algorithms with advice by
Boyar et al.~\cite{boyar2017online}. For a complete and rigorous introduction
to online problems we refer to the book by Komm \cite{komm2016introduction}. 

In a \emph{node-deletion problem} on a graph, the objective is to delete a minimal
number of vertices such that a graph satisfies a certain property. Many graph
problems are (or can be viewed as) node-deletion problems. The corresponding
properties are, for example, cycle-free graphs (\textsc{Feedback Vertex Set}),
edge-free graphs (\textsc{Vertex Cover}), or complete graphs
(\textsc{Max-Clique}).  Yannakakis~\cite{yannakakis1978node} showed that this
kind of problem is NP-complete for general, non-trivial hereditary graph
properties. The advice complexity of  the corresponding online problems was
then studied by Komm et al.\ \cite{komm2016advice} using results of Boyar et
al.\ \cite{boyar2015advice}. In this paper, we study deterministic algorithms
for online induced forbidden subgraph problems. The input graph is revealed
vertex by vertex, along with its corresponding induced edges, and an
algorithm has to keep the graph free of any induced subgraph isomorphic to a
fixed forbidden subgraph $H$ by deleting vertices. It seems only natural that
such an algorithm does not have to decide immediately if a revealed vertex
should be deleted or not but only has to decide which vertices to delete once
an induced copy of $H$ appears in the online graph. This approach is referred
to as the \textit{delayed decision model}, introduced by Chen et al.~\cite{chen2021online} and based on the \emph{preemptive model} used by Komm et al.~\cite{komm2016advice}.
The problem itself is known as the
\textsc{Delayed $H$-Node-Deletion Problem} and was studied by
Chen et al.~\cite{chen2021online}, and Berndt and Lotze \cite{berndt2023advice} together
with other vertex and edge deletion problems with respect to their advice
complexity. 

In this paper, we expand an idea of Chen et al.\ for algorithms without advice
to give tight bounds on the competitive ratio in the case where the connected
forbidden subgraph $H$ does not contain a pair of true twins, or does not
contain a pair of false twins.
Note that this problem does not match the classical model where an immediate
decision whether or not to delete is required after receiving each vertex. In
fact, Chen et al.\ showed that the problem does not admit
any competitive algorithm under the classical model. 

We further use advice of the form described by Emek et al.~\cite{emek2011online}
where the online algorithm is augmented by a sequence of
advice queries $u_t$ with $t = 1,2,\dots$. The query $u_t$ maps the whole
request sequence $\sigma$ to an advice $u_t(\sigma)$ of fixed size. At the
$t$-th request, the algorithm is provided with advice $u_t(\sigma)$. This model
differs from other advice models since it does not reveal the whole advice
immediately to the online algorithm but only parts of the advice together with
each request. Nevertheless, the advice oracle knows the whole input and can
optimally design advice for the algorithm accordingly. Specifically, in our
model, each time a new vertex is revealed, the algorithm has access to one
additional bit of advice that indicates whether the vertex is part of a fixed
optimal solution or not.

In traditional models analyzing advice complexity, the oracle is always correct
and infallible. In practice however, the advice will be computed under some
assumptions that are not completely reliable (e.g., by a machine learning algorithm).
If an algorithm blindly trusts the advice and it turns out
that it contains errors, the consequences
for the performance of the algorithm can be dire. Therefore, it makes sense
that an algorithm should be robust against errors in the advice. This was
studied by Lykouris and Vassilvitskii~\cite{lykouris2018competitive,lykouris2021competitive}, and
Purohit et al.~\cite{purohit2018improving} as the \emph{prediction model} or model
of \emph{machine-learned advice}. The concept was further generalized and applied to
several well-known online problems by
Angelopoulos et al.~\cite{angelopoulos2020online,angelopoulos2024online} as the model of untrusted advice.

Under this model, an algorithm should fulfill two requirements. If the advice
turns out to be correct, the algorithm should perform close to the optimal
solution. We define the \emph{consistency} $r_{\textsc{Alg}}$ of an algorithm $\textsc{Alg}$ as its
competitive ratio achieved with the best possible advice that is correct and of
the expected form. Nevertheless, if the advice is incorrect or even maliciously
designed by an adversary, the performance of the algorithm should not be
compromised too much. The algorithm should hence be robust against incorrect
advice. We define the \emph{robustness} $w_{\textsc{Alg}}$ to be the competitive ratio of
the algorithm $\textsc{Alg}$ under worst-case advice. Therefore, the performance of
algorithm $\textsc{Alg}$ working with predictions (``untrusted advice''), can be expressed as a
two-dimensional point $(r_{\textsc{Alg}},w_{\textsc{Alg}})$. An algorithm $\textsc{Alg}_1$ dominates an algorithm
$\textsc{Alg}_2$ if $r_{\textsc{Alg}_1} \le r_{\textsc{Alg}_2}$ and $w_{\textsc{Alg}_1} \le w_{\textsc{Alg}_2}$. For complicated
online problems, there might not exist a single online algorithm that
dominates all others and two algorithms can generally be incomparable. Thus, we
search for a \textit{Pareto-optimal} family of algorithms $\mathcal{A}$, i.e.,
a family of pairwise incomparable algorithms $\mathcal{A}$ such that, for every
algorithm $\textsc{Alg}'$, there exists an algorithm $\textsc{Alg} \in \mathcal{A}$ that dominates
$\textsc{Alg}'$. 

We give a family of algorithms $\textsc{Alg}_p$ with a suitable parameter $p$ that solves
the \textsc{Delayed $H$-Node-Deletion Problem} with predictions and prove
that it is Pareto-optimal for certain forbidden subgraphs $H$ when limiting the
predictions to the model we described above. Specifically, our results apply to
connected forbidden subgraphs that are 2-vertex connected but do not contain
two true twins or do not contain two false twins, as well as for forbidden
paths of a fixed size greater than four or equal to two. We also investigate
the competitive ratios of algorithms solving the
\textsc{Delayed $H$-Node-Deletion Problem} without advice. This can be motivated by the fact that
an online algorithm cannot possibly perform better on incorrect, adversarially
chosen predictions than an online algorithm that chooses to ignore the predictions or does not receive any.
A na\"{\i}ve algorithm can just delete all vertices of the induced copy of a
forbidden induced subgraph whenever it appears. We show that this strategy is
optimal with respect to the competitive ratio for connected forbidden subgraphs
that do not contain two true twins or do not contain two false twins. This
includes common subgraphs such as cliques and induced cycles, stars, or paths.
Furthermore, we propose a forbidden subgraph for which a better algorithm might
exist.

\section{Preliminaries}\label{chap:prelim}

We use standard graph notation and we consider simple and undirected
graphs only. For a given graph $G=(V,E)$, $|G|$ denotes the number of
vertices (or nodes) $|V(G)|$; $C_k$ denotes the cycle, $P_k$ the path, and
$K_k$ the complete graph consisting of $k$ vertices.
For a subset $S\subseteq V$, we define $G-S$ to be the graph $G[V\setminus S]$
induced by the deletion of all vertices $v \in S$. 

A graph $G$ is $H$-free if there is no induced copy of the subgraph $H$ in
$G$, i.e., there exists no induced subgraph isomorphic to $H$ in $G$. The
\emph{open neighborhood} $N(v)$ of vertex $v$ is the set of vertices adjacent to $v$;
the \emph{closed neighborhood} of $v$ is defined as $N[v] = N(v) \cup \{v\}$. Two
vertices are \emph{true twins} if they have the same closed neighborhood and
\emph{false twins} if they have the same open neighborhood.

An online graph $G$ is a graph that is induced by its vertices which are
revealed one by one. The set of vertices $V(G)= \{v_1,v_2,\dots,v_n\}$ is ordered
by their occurrence in the online instance. The graph $G_t$ is the graph
induced by the first $t$ vertices of online graph $G$, i.e., $G_t =
G[v_1,\dots,v_t]$. For a fixed subgraph $H$, the
\textsc{Delayed $H$-Node-Deletion Problem} on graph $G$ is to select for every $t$ with $1\le t \le n$ a set $S_t
\subseteq {V(G_t)}$ such that $G_t-S_t$ is $H$-free and
$S_1 \subseteq \dots\subseteq S_n$. The goal is to minimize the size of $S_n$. The
\textsc{Delayed Connected $H$-Node-Deletion Problem} is the
same problem for a fixed connected subgraph $H$. An online algorithm working on
this problem has to decide on $S_t$ based only on $G_t$, independently of any
vertices that are revealed afterwards. If the algorithm works with advice, it additionally has access
to the values of $u_1(G),\dots,u_t(G)$ at step $t$. In the case of predictions
(``untrusted advice''), these can be correct or incorrect. We only consider the case where
$u_t(G) \in \{0,1\}$. The value $u_t(G)$ gives advice for vertex $v_t$ of graph
$G$. If $u_t(G) = 1$, the advice suggests that the vertex $v_t$ is part of a
fixed optimal solution of the problem and should be deleted, i.e., added to the
set $S_t$. If $u_t(G) = 0$, the advice suggests that the vertex $v_t$ should
not be deleted. 

For this minimization problem, $\text{cost}(\textsc{Opt}(G))$ denoted the cost
of the optimal solution on graph $G$, i.e., the least number of vertices that
need to be deleted in order for the graph $G$ to be $H$-free. For a
deterministic online algorithm $\textsc{Alg}$, $\text{cost}( \textsc{Alg}(G))$
represents the number of vertices the algorithm $\textsc{Alg}$ deletes during
its execution on the online graph $G$ and
$\text{cost}(\textsc{Alg}^{u_t(G)}(G))$ represents the number of vertices the algorithm
$\textsc{Alg}$ deletes with access to advice values $u_t(G)$ for $1 \le t \le
n$. The algorithm $\textsc{Alg}$ is $c$-competitive if, for every online graph
$G$ and some constant non-negative $\alpha$, $\text{cost}( \textsc{Alg}(G)) \le
c \cdot \text{cost}(\textsc{Opt}(G)) + \alpha$. The competitive ratio of
$\textsc{Alg}$ is defined as $c_{\textsc{Alg}} = \inf \{ c \ge 1 \mid \textsc{Alg}
\text{ is $c$-competitive} \}$. If the algorithm $\textsc{Alg}$ works with
access to predictions, $\textsc{Alg}$ is $(r,w)$-competitive if, for every
online graph $G$, $\text{cost}( \textsc{Alg}^{u_t(G)}(G)) \le r \cdot
\text{cost}(\textsc{Opt}(G)) + \alpha$ for some correct advice $u_t(G)$ (as
defined above) and $\text{cost}( \textsc{Alg}^{u_t(G)}(G)) \le w \cdot
\text{cost}(\textsc{Opt}(G)) + \alpha$ for every possible advice $u_t(G)$,
correct or incorrect, and some constant non-negative $\alpha$. The consistency
$r_{\textsc{Alg}}$ and the robustness $w_{\textsc{Alg}}$ are the corresponding
competitive ratios.

\section{A Pareto-Optimal Algorithm}\label{sec:results_upper}

A na\"{\i}ve algorithm can solve the \textsc{Delayed $H$-Node-Deletion Problem}
without advice by deleting every vertex of an induced copy of $H$ whenever it
appears. It is $k$-competitive for $k=|H|$ and was already presented by Chen et
al.~\cite{chen2021online}.

Let us now consider a simple family of algorithms $\textsc{Alg}_p$, that solve the
\textsc{Delayed $H$-Node-Deletion Problem} with predictions and establish a
first upper bound on the optimal competitive ratios.

\begin{definition}[Algorithm \boldmath$A\scalebox{0.85}{\textit{LG}}_p$\unboldmath]\label{alg:Ap}
  The algorithm $\textsc{Alg}_p$ with parameter $p \in [0,1)$ works on an online graph $G$
  and receives an advice $u_t(G)$ each time a vertex $v_t$ is revealed for $t$ with
  $1 \le t \le |G|$. The value of $u_t(G)$ is a single bit. If a vertex has advice
  $1$, it suggests that the vertex is part of a fixed optimal solution and
  should be deleted.  $\textsc{Alg}_p$ keeps track of two counters $d$ and $e$
  which are initialized to $0$. 
 
  Whenever an intact copy of an induced subgraph $H$ appears in the online graph
  $G$, $\textsc{Alg}_p$ distinguishes between the following cases to keep the graph $H$-free: 

  \begin{itemize}  
    \item \textit{Case 1.} If the copy of $H$ does not contain any vertices
    with advice $1$, then the advice must be incorrect. Going forward, the
    algorithm $\textsc{Alg}_p$ deletes all vertices of any induced subgraph $H$ that
    appears without incrementing $e$ or $d$.
    \item \textit{Case 2.} If $e/(e+d) > p$ or $d = 0$, $\textsc{Alg}_p$ deletes all $k = |H|$ vertices of the copy of $H$ 
    and increments $d$ by $1$.
    \item \textit{Case 3.} Else, $\textsc{Alg}_p$ deletes one vertex with advice $1$ of the copy of $H$ and increments $e$ by $1$. 
    If the copy $H$ contains multiple vertices with advice $1$, it chooses the one that appeared first. 
  \end{itemize}  
    
  If multiple intact copies of $H$ appear at once in $G$, $\textsc{Alg}_p$ chooses one
  arbitrary copy of $H$ first for the above case distinction and then continues
  choosing another copy of $H$ until $G$ is $H$-free. This is done before the
  next vertex of $G$ is revealed.
\end{definition}

It is clear that $\textsc{Alg}_p$ keeps any online graph $G$ $H$-free since, in every
iteration of the above case distinction, at least one induced copy of $H$ is
destroyed and $\textsc{Alg}_p$ iterates until $G$ is $H$-free before the next vertex is
revealed. 

The parameter $p$ essentially indicates how much $\textsc{Alg}_p$ trusts the advice. The
counter $e$ tracks how often $\textsc{Alg}_p$ follows the advice, while the
counter $d$ records how often it disregards it.  Before examining the
competitiveness of $\textsc{Alg}_p$, we show that $e/(e+d)$ approximates $p$ well
enough.

\begin{lemma}\label{lem:mathT}
  Consider an arbitrary online graph $G$ 
  that requires at least one vertex deletion to become $H$-free.
  Denote the final value of $d$ after the online execution of
  $\textsc{Alg}_p$ on $G$ by $\tilde{d}$ and the final value of $e$ by $\tilde{e}$. If
  $\tilde{d} + \tilde{e} > 0$, then $\tilde{e}/(\tilde{e} + \tilde{d})  \le p + 1/(\tilde{e} + \tilde{d})$.
\end{lemma}

\begin{proof}
  We show that $\tilde{e}/(\tilde{e} + \tilde{d}) \le p +
  1/(\tilde{e} + \tilde{d})$ by proving by induction over $e' + d' > 0$
  that $e'/(e' + d') \le p + 1/(e' + d')$ holds for any values
  $e'$ and $d'$ of $e$ and $d$ during the execution of $\textsc{Alg}_p$.

  The base case of $e' + d' = 1$ is only possible after case~2 of $\textsc{Alg}_p$
  has been executed once.
  Therefore, $d' = 1$ and $e' = 0$, and it follows that $e'/(e' + d') = 0
  \le p + 1$ for any $p\in [0,1)$.

  So assume that the induction hypothesis
  $e'/(e' + d') \le p + 1/(e' + d')$ holds for some $e' + d' > 0$.
  We show that the property holds for $e' + d' + 1$ by case distinction over
  which counter has been incremented last by $\textsc{Alg}_p$.

  \begin{itemize}
    \item \textit{Case 1.} The counter $e$ has been incremented last to $e'+1$
      in case~3 of the algorithm. Therefore, $e'/(e' + d') \le p$. Then it follows directly that
      \[\frac{e'+1}{e' + d' + 1}\leq \frac{e'}{e' + d' }+\frac{1}{e'+d'+1}\leq p+\frac{1}{e'+d'+1}\,.\]
    \item \textit{Case 2.} The counter $d$ has been incremented last to $d'+1$
      in Case $2$ of the algorithm. Then, $e'/(e' + d') \le p + 1/(e' + d')$ holds by induction hypothesis.
      If $e'=0$, then $e'/(e'+d'+1)\le p+1/(e'+d'+1)$ trivially. If $e'\geq 1$, then
      \begin{align*}
      \frac{e'}{e'+d'+1}&\leq \frac{e'}{e'+d'}-e'\cdot \left(\frac{1}{e'+d'}-\frac{1}{e'+d'+1}\right)\\
      &\leq p+\frac{1}{e'+d'}-\left(\frac{1}{e'+d'}-\frac{1}{e'+d'+1}\right)\\
      &=p+\frac{1}{e'+d'+1}\,.
      \end{align*}
    \end{itemize}
\end{proof}

\begin{lemma}\label{lem:mathU} 
  Consider an arbitrary online graph $G$ 
  that requires at least one vertex deletion to become $H$-free.
  Denote the final value of $d$ after the online execution of
  $\textsc{Alg}_p$ on $G$ by $\tilde{d}$ and the final value of $e$ by $\tilde{e}$. If
  $\tilde{d} + \tilde{e} > 0$, then $\tilde{d}/(\tilde{e} + \tilde{d}) \le (1-p) + 1/(\tilde{e} + \tilde{d})$.
\end{lemma}

\begin{proof}
    This proof works similar to the proof of \cref{lem:mathT}.
    The base case $e' + d' = 1$ is again only possible after case~2 of $\textsc{Alg}_p$
    has been executed once.
    Therefore, $d' = 1$ and $e' = 0$, it follows that $d'/(e' + d') = 1
    \le 2-p$ for any $p\in [0,1)$.

    So assume that the induction hypothesis $d'/(e' + d') \le (1-p) + 1/(e' + d')$
    holds for some $e' + d' > 0$.
    We show that the property holds for $e' + d' + 1$ by case distinction over
    which counter has been incremented last by $\textsc{Alg}_p$.
   
    \begin{itemize}
    \item \textit{Case 1.}
       The counter $d$ has been incremented last to $d'+1$ in Case $2$ of the
       algorithm. Therefore, $e'/(e'+d') > p$, and thus $d'/(e'+d') < 1-p$. 
       It follows directly that 
       \[\frac{d'+1}{e'+d'+1}\leq \frac{d'}{e'+d'}+\frac{1}{e'+d'+1}\leq 1-p+\frac{1}{e'+d'+1}\,.\]
       
    \item \textit{Case 2.} The counter $e$ has been incremented last to $e'+1$
      in case~3 of the algorithm. Since $e$ is only ever incremented if $d>0$, we know that $d'\geq 1$. This implies that
      \begin{align*}
      \frac{d'}{d'+e'+1}&=\frac{d'}{d'+e'}-d'\cdot\left(\frac{1}{d'+e'}-\frac{1}{d'+e'+1}\right)\\
      &\leq 1-p+\frac{1}{e'+d'}-\left(\frac{1}{d'+e'}-\frac{1}{d'+e'+1}\right)\\
      &=1-p+\frac{1}{e'+d'+1}\,.
      \end{align*}
    \end{itemize} 
\end{proof}

\begin{theorem}\label{thm:compAlg}
  $\textsc{Alg}_p$ is $(k-p\cdot (k-1), k+p/(1-p))$-competitive for $k = |H|$.
\end{theorem}

\begin{proof}
  Consider an arbitrary online graph $G$ 
  that requires $i>0$ vertex deletions to become $H$-free, i.e., $\text{cost}(\textsc{Opt}(G)) = i$.
  We denote the final value of $d$ after the online execution of $\textsc{Alg}_p$ on $G$
  by $\tilde{d}$ and the final value of $e$ by $\tilde{e}$.
  
  If the advice $u_t(G)$ is correct, $G$ contains $i$ \emph{optimal vertices} which have
  advice $1$. There are no other vertices with advice $1$ in $G$.  Every
  induced copy of $H$ in $G$ contains at least one vertex with advice $1$.
  $\textsc{Alg}_p$ deletes at most those $i$ vertices with advice $1$ 
  and at most $\tilde{d}(k-1)$ vertices with advice $0$, because it only deletes
  vertices with advice $0$ of one induced copy of $H$ when incrementing $d$, and
  there are at most $k-1$ vertices with advice $0$ per copy of $H$. We know
  that $\tilde{d} + \tilde{e} \le i$ since, every time $d$ or $e$ are
  incremented in $\textsc{Alg}_p$, at least one of the $i$ vertices with advice $1$ is
  deleted. Since $i>0$ and since the advice $u_t(G)$ is correct, $d$ is
  incremented at least once by $\textsc{Alg}_p$ and $\tilde{d} + \tilde{e} > 0$. From
  \cref{lem:mathU}, it follows that $\tilde{d} \le (1-p) (\tilde{e} +
  \tilde{d})+1 \le i(1-p)+1$, and we therefore get
  \begin{align*}
    \text{cost}(\textsc{Alg}_p^{u_t(G)}(G)) &\le i+\tilde{d}(k-1) \\
    &\le i+i(1-p)(k-1)+(k-1) \\
    &\le (k-p\cdot (k-1)) \cdot \text{cost}(\textsc{Opt}(G)) + (k-1) \,,
  \end{align*}
  and thus $r_{\textsc{Alg}_p} \le k-p\cdot (k-1)$ as $k-1$ is constant.
  
  If the advice $u_t(G)$ is incorrect but $G$ does not contain any induced
  subgraph $H$ without a vertex with advice $1$, then $\textsc{Alg}_p$ deletes exactly
  $\tilde{e}+\tilde{d}k$ vertices.  $\textsc{Alg}_p$ guarantees that every time $d$ is
  incremented, at least one out of the $i$ optimal vertices from a distinct
  $\textsc{Opt}(G)$ of $G$ is deleted, because every induced copy of $H$ in $G$
  must include at least one of those, and every time $d$ is incremented, all
  vertices of an induced $H$-subgraph are deleted. Therefore, $\tilde{d} \le i$
  holds.
  Since $i>0$, $d$ is incremented at least once by $\textsc{Alg}_p$ and $\tilde{d} +
  \tilde{e} > 0$. From \cref{lem:mathT}, it follows that
  $\tilde{e} \le p(\tilde{e}+\tilde{d})+1$, which implies
  \[ \tilde{e}  \le \frac{p \tilde{d} + 1}{1-p} \le \frac{p i + 1}{1-p}  \,. \]
  Hence, we get
  \[
    \text{cost}(\textsc{Alg}_p^{u_t(G)}(G)) \le \tilde{e}+\tilde{d}k 
    \le \frac{pi+1}{1-p} + ik 
    \le \left(k+ \frac{p}{1-p}\right) \cdot \text{cost}(\textsc{Opt}(G)) + \frac{1}{1-p}  \,.
  \]

  If the advice $u_t(G)$ is incorrect and contains an induced copy of $H$ without any vertex with advice $1$, 
  $\textsc{Alg}_p$ still deletes $\tilde{e}+\tilde{d}k$ vertices until the copy of $H$
  without any vertex with advice $1$ appears.  After that it deletes $v$ times
  $k$ vertices for some $v > 0$. Note that, after the first copy of $H$ without
  any vertex with advice $1$ appeared, $e$ and $d$ are not further incremented
  and reached their final value $\tilde{e}$ and $\tilde{d}$. The implications
  of \cref{lem:mathT,lem:mathU} still hold. It further holds that
  $\tilde{d} + v \le i$, 
  because every time $d$ is incremented and for every one of the $v$ deletions, $\textsc{Alg}_p$ deletes
  at least one of the $i$ optimal vertices of a distinct $\textsc{Opt}(G)$ of
  $G$ since $\textsc{Alg}_p$ deletes all vertices of the appearing $H$-subgraphs. If
  $\tilde{d} + \tilde{e} > 0$, then $\tilde{e} \le (pi+1)/(1-p)$ still
  holds by \cref{lem:mathT}. Otherwise, $\tilde{e}=0 \le (pi+1)/(1-p)$
  holds trivially. We get
  \[
    \text{cost}(\textsc{Alg}_p^{u_t(G)}(G)) \le \tilde{e}+(\tilde{d}+v)k 
    \le \frac{pi+1}{1-p} + ik 
    \le \left(k+ \frac{p}{1-p} \right) \cdot \text{cost}(\textsc{Opt}(G)) + \frac{1}{1-p} \,.
  \]
  It follows that $w_{\textsc{Alg}_p} \le k+ p/(1-p) $ holds since $1/(1-p)$ is constant for a fixed $\textsc{Alg}_p$.
\end{proof}
  
Note that $\textsc{Alg}_p$ is not defined for $p=1$, but we can easily define an Algorithm $\textsc{Alg}_1$ which always trusts the 
advice and only deletes vertices with advice $1$ if possible. If the advice is trusted, $\textsc{Alg}_1$ deletes an 
optimal solution and is $1$-competitive. But if the advice is untrusted, there might exist
an online graph $G$ with adversarially chosen advice on which $\textsc{Alg}_1$ would perform arbitrarily badly.
For this, consider a forbidden connected subgraph $H$ with  $|H| > 1$ and an online graph $G$ that presents arbitrarily many 
induced copies of $H$ which all overlap at exactly one vertex and are otherwise disjoint such that deleting this vertex would result in an 
$H$-free graph. The optimal solution would thus require one deletion. But in the online problem with 
predictions, the adversary could now choose advice $0$ for this optimal vertex and advice $1$ for some other, 
non-optimal vertex for each appearing induced copy of $H$. $\textsc{Alg}_1$ deletes the vertices with advice $1$ for each copy of $H$ and since there 
is no copy of $H$ without any vertex with advice $1$, $\textsc{Alg}_1$ never deletes the optimal vertex. 
For arbitrarily many induced copies of $H$ presented in such a way, $\textsc{Alg}_1$ performs arbitrarily badly.

While algorithm $\textsc{Alg}_p$ (\cref{alg:Ap}) even works for unconnected $H$, we now
focus on connected forbidden subgraphs $H$ and hence on the
\textsc{Delayed Connected $H$-Node-Deletion Problem}.

\section{Lower Bounds without Advice} \label{sec:results_wo_advice}

We now inspect the competitive ratios of algorithms solving the \textsc{Delayed
Connected $H$-Node-Deletion Problem} without advice for connected induced
forbidden subgraphs $H$. Chen et al.~\cite{chen2021online} showed that there
is no online algorithm without advice solving the \textsc{Delayed
$H$-Node-Deletion Problem} for the forbidden subgraph $H=C_k$ (i.e., the cycle on
$k$ vertices) with a competitive ratio better than $k$, for any $k>4$, by
giving adversarial strategies that force $k$ deletions for a gadget that
requires $1$ deletion in the offline setting. We use and expand their idea for
more general $H$.

\begin{lemma} \label{lem:unadv_cl}
  Let $H$ be a connected subgraph that does not contain two false twins.
  There does not exist any deterministic algorithm solving the
  \textsc{Delayed Connected $H$-Node-Deletion Problem} with a competitive ratio
  better than $k = |H| >1$.
\end{lemma}

\begin{proof}
  Recall that two vertices are false twins if they are non-adjacent and have
  the same open neighborhood. An adversary can construct $k$ different gadgets
  $g^i$ for $i \in [1,k]$ in the following way. First, a copy of $H$ is
  presented, formed by vertices $v_1, \dots ,v_k$. For every deleted vertex
  except vertex $v_i$, a new vertex with the same open neighborhood as the
  deleted vertex is reinserted. Apart from that, no additional edge, in
  particular no edge between $v$ and $v'$, is introduced. It is clear that each
  of these reinsertions produces a new induced copy of $H$ in $g^i$, thus
  forcing another vertex deletion to keep the gadget $H$-free until $v_i$ is
  deleted. Since $v_i$ can be chosen arbitrarily and is indistinguishable from
  all remaining vertices of the originally presented subgraph $H$, there is,
  for any deterministic algorithm, a gadget $g^i$ for which the algorithm has
  to delete every vertex of the originally presented subgraph $H$ until it
  deletes $v_i$ (if ever). Thus, the gadget is forcing it to delete at least
  $k$ vertices until $g^i$ is finally $H$-free.
  
  We now prove that only deleting $v_i$ would have sufficed, i.e., $g^i - \{
  v_i \}$ is $H$-free for any $g^i$ formed by above construction under an
  arbitrary online algorithm. We partition the set $V^i$ of vertices of $g^i$
  in $k$ equivalence classes $V^i_1, \dots ,V^i_k$, such that $v \in V^i_j$ if
  $v=v_j$ or $v$ was reinserted for some $v'$ and $v' \in V^i_j$ for $j \in
  [1,k]$; see \cref{fig:reinsertionK3} for an illustration. Note that by
  construction of the reinsertions, if $v,u \in V^i_j$ and $v \neq u$, then
  $v$ and $u$ are non-adjacent. Also, $V^i_i = \{ v_i \}$ since there are no
  reinsertions for $v_i$. Therefore, $V^i \backslash \{ v_i \}$ is
  partitioned into $k-1$ equivalence classes with $V^i \backslash \{ v_i \} =
  \bigcupdot_{i \neq j} V^i_j$ for $j \in [1,k]$. Any subset of $k$ vertices
  of $V^i\backslash \{ v_i \}$ must hence include at least two vertices from
  the same equivalence class that are by construction non-adjacent and share
  the same neighborhood and hence are false twins. Thus, there cannot be an
  induced subgraph $H$ in $g^i - \{ v_i \}$ and $v_i$ forms an optimal
  solution of size one for $g^i$. Note that, if an algorithm chooses to never
  delete $v_i$, it has an arbitrarily bad competitive ratio because it
  deletes arbitrary many vertices on $g^i$, where one would suffice.

  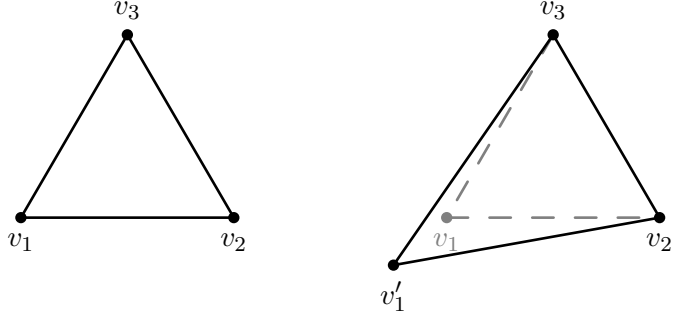
\begin{figure}[t]
    \begin{center}
      \definecolor{yqyqyq}{rgb}{0.5019607843137255,0.5019607843137255,0.5019607843137255}
      \begin{tikzpicture}[line cap=round,line join=round,x=0.7cm,y=0.7cm]
      	\coordinate (a) at (1,-0.9);
      	\def\radius{1.5pt};
        \draw [line width=1.pt,dash pattern=on 8pt off 8pt,color=yqyqyq] (4.,3.4641016151377544)-- (2.,0.);
        \draw [line width=1.pt,dash pattern=on 8pt off 8pt,color=yqyqyq] (2.,0.)-- (6.,0.);
        \draw [line width=1.pt] (4.,3.4641016151377544)-- (6.,0.);
        \draw [line width=1.pt] (-4.,3.4641016151377544)-- (-6.,0.);
        \draw [line width=1.pt] (-6.,0.)-- (-2.,0.);
        \draw [line width=1.pt] (-4.,3.4641016151377544)-- (-2.,0.);
        \draw [line width=1.pt] (4.,3.4641016151377544)-- (a);
        \draw [line width=1.pt] (a)-- (6.,0.);
        \node[fill=black, label=below:{$v_2$}, circle, inner sep=\radius] at (6.,0.) {};
        \node[color=yqyqyq, fill, label={[text=yqyqyq]below:{$v_1$}}, circle, inner sep=\radius] at (2,0) {};
        \node[fill=black, label=above:{$v_3$}, circle, inner sep=\radius] at (4.,3.4641016151377544) {};
        \node[fill=black, label=below:{$v_2$}, circle, inner sep=\radius] at (-2.,0.) {};
        \node[fill=black, label=below:{$v_1$}, circle, inner sep=\radius] at (-6.,0.) {};
        \node[fill=black, label=above:{$v_3$}, circle, inner sep=\radius] at (-4.,3.4641016151377544) {};
        \node[fill=black, label=below:{$v_1'$}, circle, inner sep=\radius] at (a) {};
      \end{tikzpicture}
    \end{center}
    \caption{A gadget $g^i$ according to \cref{lem:unadv_cl} for $H = K_3$ and $i \neq 1$. On the left after presenting the first copy of $H$ and on the right after an algorithm deleted $v_1$ and $v_1'$ was reinserted. The vertices $v_1$ and $v_1'$ are in the same equivalence class $V^i_1$. Grey vertices indicate that they were deleted by the algorithm and the incident edges of deleted vertices are displayed as dashed.}
    \label{fig:reinsertionK3}
  \end{figure}

  For any deterministic online algorithm $\textsc{Alg}$ that solves this
  problem and arbitrary $m \ge 1$, there exists now an adversarial strategy
  which presents an online graph $G^m$ that repeats $m$ such vertex-disjoint
  gadgets $g^{i_1}, \dots, g^{i_m}$ such that it forces at least $k$ vertex
  deletions for each gadget where one would suffice by always choosing an $i_l$
  such that $v_{i_l}$ is deleted last by $\textsc{Alg}$. Therefore,
  $\text{cost}(\textsc{Alg}(G^m)) \ge mk$ and $\text{cost}(\textsc{Opt}(G^m)) =
  m$. It follows that there does not exist any deterministic online algorithm
  solving this problem with a competitive ratio better than $k$. 
\end{proof}

We can show the same property for forbidden subgraphs $H$ that do not contain
true twins, i.e., two adjacent vertices with the same closed neighborhood.

\begin{lemma} \label{lem:unadv_tr}
   Let $H$ be a connected subgraph that does not contain two true twins. There
   does not exist any deterministic algorithm solving the \textsc{Delayed
   Connected $H$-Node-Deletion Problem} with a competitive ratio better than $k =
   |H| > 1$. 
\end{lemma}

\begin{proof}
  We prove the claim by using a similar gadget $g^i$ for $i \in [1,k]$ as
  before. First, a copy of $H$ formed by vertices $v_1, \dots ,v_k$ is
  presented. We again define reinsertions and a partition of the set $V^i$
  of vertices of $g^i$ into $k$ equivalence classes $V^i_1, \dots ,V^i_k$
  such that $v \in V^i_j$ if $v=v_j$ or $v$ was reinserted for some $v'$
  and $v' \in V^i_j$ for $j \in [1,k]$. Whenever a vertex $v \neq v_i$ is
  deleted, a new vertex $v'$ is reinserted. Vertex $v'$ shares an edge with
  every vertex of the neighborhood of $v$ and with every vertex in its
  equivalence class. Because $v$ and $v'$ are in the same equivalence
  class, the reinsertion leads to them having the same closed neighborhood;
  see \cref{fig:reinsertionElse} for an illustration. 
  \begin{figure}[t]
      \begin{center}
        \definecolor{yqyqyq}{rgb}{0.5019607843137255,0.5019607843137255,0.5019607843137255}
          \begin{tikzpicture}[line cap=round,line join=round,x=0.7cm,y=0.7cm]
          	  \def\radius{1.5pt};
          	  \coordinate (a) at (1,-1);
              \draw [line width=1.pt] (-2.,0.)-- (-6.,0.);
              \draw [line width=1.pt] (-6.,0.)-- (-6.,4.);
              \draw [line width=1.pt] (-6.,4.)-- (-2.,4.);
              \draw [line width=1.pt] (-2.,4.)-- (-2.,0.);
              \draw [line width=1.pt,dash pattern=on 8pt off 8pt,color=yqyqyq] (2.,0.)-- (6.,0.);
              \draw [line width=1.pt] (6.,0.)-- (6.,4.);
              \draw [line width=1.pt] (6.,4.)-- (2.,4.);
              \draw [line width=1.pt,dash pattern=on 8pt off 8pt,color=yqyqyq] (2.,4.)-- (2.,0.);
              \draw [line width=1.pt,dash pattern=on 8pt off 8pt,color=yqyqyq] (a)-- (2.,0.);
              \draw [line width=1.pt] (2.,4.)-- (a);
              \draw [line width=1.pt] (a)-- (6.,0.);
              \node[fill=black, label=below:{$v_2$}, circle, inner sep=\radius] at (-2.,0.) {};
              \node[fill=black, label=below:{$v_1$}, circle, inner sep=\radius] at (-6.,0.) {};
              \node[fill=black, label=above:{$v_4$}, circle, inner sep=\radius] at (-6.,4.) {};
              \node[fill=black, label=above:{$v_3$}, circle, inner sep=\radius] at (-2,4.) {};

              \node[color=yqyqyq, fill,label={[text=yqyqyq]below right:{$v_1$}}, circle, inner sep=\radius] at (2,0) {};
              \node[fill=black, label=below:{$v_2$}, circle, inner sep=\radius] at (6,0) {}; 
              \node[fill=black, label=above:{$v_3$}, circle, inner sep=\radius] at (6,4.) {};
              \node[fill=black, label=above:{$v_4$}, circle, inner sep=\radius] at (2,4.) {};
              \node[fill=black, label=below:{$v_1'$}, circle, inner sep=\radius] at (a) {};
          \end{tikzpicture}
      \end{center}
      \caption{A gadget $g^i$ according to \cref{lem:unadv_tr} for $H = C_4$ and $i \neq 1$. On the left after the presentation of the first copy of $H$ and on the right after an algorithm deleted $v_1$ and $v_1'$ was reinserted. The vertices $v_1$ and $v_1'$ are in the same equivalence class $V^i_1$.}
      \label{fig:reinsertionElse}
  \end{figure}
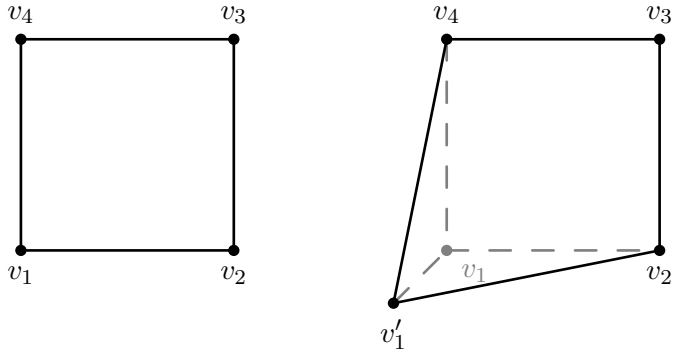
 
  The remaining part of the proof works similar to the proof of
  \cref{lem:unadv_cl}. In particular, since every reinsertion produces a
  new induced subgraph $H$ in $g^i$, any deterministic algorithm can be
  forced to delete at least $k$ vertices until $g^i$ is $H$-free for some
  gadget $g^i$. Since $V^i_i = \{ v_i \}$ holds, any subset of $k$
  connected vertices of $V^i\backslash \{ v_i \}$  must include at least
  two vertices $v$ and $v'$ from the same equivalence class, which are by
  construction adjacent and share the same closed neighborhood and are thus
  true twins. Because subgraph $H$ consists of $k$ connected vertices and
  does not include true twins, there cannot be an induced copy of $H$ in
  $g^i - \{ v_i \}$. Therefore, only one vertex deletion is required to get
  any gadget $g^i$ $H$-free and vertex $v_i$ is an optimal solution. 

  By combining $m$ such vertex-disjoint gadgets to a graph $G^m$, we force
  any deterministic algorithm $\textsc{Alg}$ to have
  $\text{cost}(\textsc{Alg}(G^m)) \ge mk$ and
  $\text{cost}(\textsc{Opt}(G^m)) = m$. It follows that there does not
  exist any deterministic online algorithm solving this problem with a
  competitive ratio better than $k$.
\end{proof}

Let us note some properties of these gadgets that directly follow from our
considerations and which we will further use in the proofs of
\cref{thm:boundLower,thm:boundLowerTight}.

\begin{remark} \label{rem:properties}
  Let $H$ be a connected subgraph that does not contain two true twins or does
  not contain two false twins. For any deterministic algorithm $\textsc{Alg}$
  that solves the \textsc{Delayed Connected $H$-Node-Deletion Problem} and is not
  arbitrarily bad, there exists a gadget $g^i$ described in \cref{lem:unadv_cl}
  (resp.\ \cref{lem:unadv_tr}) such that $\textsc{Alg}$ has to eventually
  delete $v_i$ in such a way that it has to delete all $k$ vertices of the copy
  of $H$ originally presented in $g^i$. Furthermore, every induced copy of $H$
  in $g^i$ must consist of exactly one vertex out of each of the $k$
  equivalence classes and during the execution of the algorithm $\textsc{Alg}$
  on $g^i$, there exists at most one undeleted vertex of each equivalence
  class. The adversary can choose $m > 0$ such vertex-disjoint gadgets that
  form an online graph $G$ and force at least $m \cdot k$ vertex deletions
  where $m$ deletions would suffice. 
\end{remark}

The above proofs of \cref{lem:unadv_cl,lem:unadv_tr} introduce two
ways to construct powerful gadgets $g^i$ which we continue to use in this
paper. Some common subgraphs $H$ which are covered by those lemmas and hence do
not yield an algorithm that is better than $k$-competitive for the
\textsc{Delayed Connected $H$-Node-Deletion Problem} are cliques $K_k$, induced
cycles $C_k$, induced stars $S_k$ or induced paths $P_k$. Also included are
triangle-free connected subgraphs $H$ with $H \ge 3$, i.e., subgraphs that do
not contain an induced triangle, because true twins form an induced triangle
with any of their neighbors and thus cannot exists in triangle-free connected
subgraphs with more than two vertices.

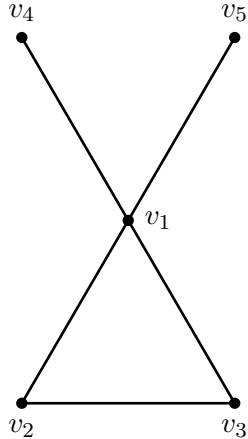
\begin{figure}[t]
  \begin{center}
  \definecolor{yqyqyq}{rgb}{0.5019607843137255,0.5019607843137255,0.5019607843137255}
  \begin{tikzpicture}[line cap=round,line join=round,x=0.7cm,y=0.7cm]
  		\def\radius{1.5pt};
        \draw[line width=1.pt] (0.,3.4641016151377544)-- (-2.,0.);
        \draw[line width=1.pt] (2.,0.)-- (-2.,0.);
        \draw[line width=1.pt] (0.,3.4641016151377544)-- (2.,0.);
        \draw[line width=1.pt] (-2.,6.928203230275509)-- (0.,3.4641016151377544);
        \draw[line width=1.pt] (0.,3.4641016151377544)-- (2.,6.928203230275509);
        \node[fill=black, label=below:{$v_2$}, circle, inner sep=\radius] at (-2,0) {};
        \node[fill=black, label=below:{$v_3$}, circle, inner sep=\radius] at (2,0) {};
        \node[fill=black, label=right:{$v_1$}, circle, inner sep=\radius] at (0.,3.4641016151377544) {};
        \node[fill=black, label=above:{$v_4$}, circle, inner sep=\radius] at (-2.,6.928203230275509) {};
        \node[fill=black, label=above:{$v_5$}, circle, inner sep=\radius] at  (2.,6.928203230275509) {};
  \end{tikzpicture}
  \end{center}
  \caption{A subgraph for which an algorithm with a better competitive ratio than $k=5$ might exist, that solves the \textsc{Delayed Connected $H$-Node-Deletion Problem} without advice.}
  \label{fig:windmill}
\end{figure}

Let us shortly investigate a connected subgraph $H$ which has both a pair of
true and a pair of false twins and might yield an algorithm with a competitive
ratio better than $k$ for $k=|H|$. Consider the subgraph $H$ with $k=5$
consisting of an induced triangle and two vertices with degree one which are
both adjacent to the same vertex of the triangle; see \cref{fig:windmill}. An
adversary could try to combine both kinds of gadgets to still force $k$ vertex
deletions where one optimal deletion suffices. But we can easily show that this
is not possible. To have one optimal vertex $v_o$ in a graph $G$, i.e.,
$G-{v_o}$ is $H$-free and $G$ is not, this vertex $v_o$ must be part of every
induced subgraph $H$ in $G$. Consider an algorithm that deletes one vertex for
every intact induced copy of $H$ that appears in $G$ and chooses, if possible,
the vertex to delete in such a way that it is part of all (or as many as
possible) induced subgraphs $H$ in $G$, even if they are already destroyed.
Such a vertex always exists if we limit $G$ to have one optimal vertex and the
algorithm therefore deletes only potentially optimal vertices. Thus, the copy
of $H$, which is first presented by the adversary, must be presented in such a
way, that every vertex of this subgraph $H$ could be the optimal, if the
adversary wants to enforce $k$ vertex deletions. Let us denote the vertices of
this copy by $v_1$, $v_2$ and $v_3$ as the vertices of the triangle, and $v_4$
and $v_5$ as the vertices with degree one, which are both adjacent to vertex
$v_1$; see \cref{fig:windmill}. If the algorithm now chooses to delete $v_4$,
the adversary must introduce a new intact copy of $H$ which includes the four
other undeleted vertices of the original copy to keep all of them potentially
optimal. There are two ways to do this.  To create a new induced copy of $H$, a
reinserted vertex $v_4'$ must be adjacent to $v_1$ and cannot be adjacent to
any other of the four undeleted vertices of the subgraph $H$. If $v_4'$ is also
adjacent to $v_4$, there exists a new induced subgraph $H$ formed by vertices
$v_4$, $v_4'$, $v_1$, $v_2$ and $v_5$. Because $v_3$ is not included in this
copy, it cannot be optimal and there are less than four undeleted potentially
optimal vertices left (in fact only two since the same happens to $v_2$). Thus,
the adversary cannot enforce a total of $k$ vertex deletions anymore. If
$v_4'$ is non-adjacent to $v_4$, there exists a new induced subgraph $H$ formed
by vertices $v_4$, $v_4'$, $v_1$, $v_2$ and $v_3$. Vertex $v_5$ is not included
and hence cannot be optimal. Note that we showed that there does not exist a
graph enforcing $k$ vertex deletion that has an optimal solution of size one.
This does not necessarily imply that there is an algorithm with a better
competitive ratio than $k$ for this $H$, but a gadget enforcing this might be
much more complex than the gadgets we considered.

\section{Lower Bounds with Predictions} \label{sec:results_w_advice}

We can apply the results we got in \cref{sec:results_wo_advice} to our model with predictions. For this, we inspect the competitive ratios of algorithms
solving the \textsc{Delayed Connected $H$-Node-Deletion Problem} with predictions on the gadgets used in \cref{lem:unadv_cl,lem:unadv_tr}.
Furthermore, we give a class of induced subgraphs $H$ on which the family of
algorithms $\textsc{Alg}_p$ (\cref{alg:Ap}) is Pareto-optimal in \cref{thm:boundLowerTight}.

First, we investigate whether a sophisticated online algorithm could
distinguish between correct and incorrect advice early on.

\begin{lemma}\label{lem:tr-untr} 
  Consider an online graph $G$ with $|G| = n$ and a sequence of untrusted
  advice queries $u_t$ for the \textsc{Delayed Connected $H$-Node-Deletion
  Problem} and some fixed $H$ with $|H| >1$. If for each induced copy of $H$
  in $G$ at least one vertex has advice $1$, then $G$ can be expanded to an
  online graph $G'$ such that $G'_n = G$ and the same advice queries $u_t$
  deliver advice $u_t(G')$ for $G'$ where $u_t(G')$ is the only correct
  advice for $G'$, the advice for the $n$ vertices of $G$ is equal to the
  advice for the first $n$ vertices of $G'$, and there are no additional
  vertices with advice $1$ added in $G'$, i.e., $u_t(G')= u_t(G)$ for $t \in
  [1,n]$ and $u_t(G')= 0$ for $t>n$.
\end{lemma}

\begin{proof}
  Recall that the value of $u_t(G)$ delivers one bit of advice for each
  vertex $v_t$ of $G$ with $t \in [1,n]$. If $u_t(G) =1$, then the advice
  suggests that $v_t$ is part of a fixed optimal solution and should be
  deleted. Assume that every induced copy of $H$ in $G$ contains at least one
  vertex $v_t$ with advice $1$, i.e., $u_t(G) = 1$. Let $i$ denote the total
  number of vertices with advice $1$ in $G$ and $k = |H|$. We show that there
  exists a graph $G'$ with $G'_n = G$ and $|G'| = n + 2i (k-1)$ for which the
  advice $u_t(G')$ with $t \in [1,n + 2i (k-1)]$ is correct and encodes a
  fixed optimal solution. Recall that $G'_n$ is the graph induced by the
  first $n$ vertices that occur in $G'$. To construct $G'$ we expand $G$ by
  presenting $k-1$ new vertices with advice $0$ for each of the $i$ vertices
  with advice $1$ in $G$ so that they form an induced copy of $H$ together
  with the vertex with advice $1$ of $G$ and are disjoint to the rest of the
  graph. We do this twice such that each of the $i$ vertices with advice $1$
  is part of two otherwise disjoint induced copies of $H$ with $2(k-1)$ newly
  introduced vertices. The graph $G'$ consists of $G$ and these $2i (k-1)$
  new vertices. The advice queries $u_t$ deliver equal advice for the
  vertices of $G$ and $G'_n$ and advice $0$ for the newly added vertices of
  $G'$, i.e., $u_t(G')= u_t(G)$ for $t \in [1,n]$ and $u_t(G')= 0$ for $t>n$.
  It follows that there are at least $i > 0$ vertex-disjoint induced
  $H$-subgraphs in $G'$ and thus $\text{cost}(\textsc{Opt}(G')) \ge i$. The
  $i$ vertices with advice $1$ therefore form an optimal solution for $G'$,
  if deleting those makes $G'$ $H$-free. To show this, assume that there is
  an induced copy of $H$ in $G'$ after deleting all vertices with advice
  $1$ to show a contradiction. After the deletions, there cannot be an
  induced subgraph $H$ that contains any of the $2i (k-1)$ newly introduced
  vertices because each of those are, after the deletion of the vertices
  with advice $1$, in a connected component together with at most $k-2$
  other vertices. Therefore, there must be an induced subgraph $H$ formed
  by the original vertices of $G$ that does not contain any vertex with
  advice $1$. But there is by assumption no induced copy of $H$ in $G$ that
  contains no vertex with advice $1$, which leads to a contradiction.
  Hence, the $i$ vertices with advice $1$ form an optimal solution for $G'$
  and advice $u_t(G')$ is correct. To show that $u_t(G')$ is the only
  correct advice for $G'$ we show that the optimal solution formed by the
  vertices with advice $1$ is unique. For that, assume that there is an
  optimal solution of size $i$ which does not include some vertex $v$ with
  advice $1$. The $2(k-1)$ newly introduced vertices for $v$ form two
  otherwise disjoint induced subgraphs $H$ together with $v$. If $v$ is not
  part of an optimal solution, there must be two other vertices out of
  those $2(k-1)$ newly introduced vertices with advice $0$ that are part of
  this optimal solution. But deleting $v$ instead of those two vertices
  would suffice because the newly introduced vertices are not part of any
  induced subgraph $H$ after deleting $v$ as shown above. Therefore, we
  would get a smaller optimal solution which is a contradiction. This
  concludes the proof.
\end{proof}

\begin{theorem} \label{thm:boundLower}
  Let $H$ be a connected subgraph that does not contain two true twins or
  does not contain two false twins. Then, there does not exist any
  deterministic algorithm solving the \textsc{Delayed Connected
  $H$-Node-Deletion Problem} with predictions that has both a consistency
  of less than $k-1/(2-p) \cdot (k-1)$ and a robustness of less than
  $k+p/(1-p)$ for any $p \in [0,1)$ and $k = |H| > 1$.
\end{theorem}

\begin{proof}
  Consider an arbitrary algorithm $\textsc{Alg}$ solving the \textsc{Delayed
  Connected $H$-Node-Deletion Problem} with predictions that is not
  arbitrarily bad. An adversary constructs the following graph $G$ and advice
  $u_t(G)$ which gives bitwise advice for each vertex of $G$. Consider the
  gadget $g^i$ that we used in the proof of \cref{lem:unadv_cl} or, if subgraph
  $H$ does not contain true twins, of \cref{lem:unadv_tr} respectively. The
  adversary chooses advice for each vertex $v \in V^i$ for such $g^i$ by
  assigning advice bit $1$ to $v$ if $v \in V^i_1$ and advice $0$ to $v$
  otherwise. Recall that $V^i_j$ is the equivalence class such that $v \in
  V^i_j$ if $v=v_j$ or $v$ was reinserted for some $v'$ and $v' \in V^i_j$.
  From the properties listed in \cref{rem:properties}, exactly one vertex for
  each induced subgraph $H$ in $g^i$ has advice $1$ and there is always at most
  one undeleted vertex with advice $1$ in $g^i$. The properties of
  \cref{rem:properties} still hold on $\textsc{Alg}$ with such potentially
  incorrect advice because the advice could be easily constructed by an
  algorithm without advice since it does not reveal any additional information.
  In particular, the adversary can choose $m > 0$ such vertex-disjoint gadgets
  that form an online graph $G$ and force at least $m \cdot k$ vertex deletions
  where $m$ deletions would suffice.
  
  Advice $u_t(G)$ delivers advice for each vertex of each gadget as described
  above. According to \cref{lem:tr-untr}, the adversary can now expand online
  graph $G$ to an online graph $G'$ such that the graph $G'$ reveals $G$ first
  with the same advice for each vertex and has the advice $u_t(G')$ as unique
  correct advice. This is possible because for every disjoint gadget in $G$,
  there is a vertex with advice $1$ in each induced copy of $H$ as shown above.
  There are, by construction, no other induced copies of $H$ in $G$. Algorithm
  $\textsc{Alg}$ performs the same deletions on the vertices of $G$,
  independent of whether $G$ is expanded to $G'$ or not, because it operates on
  the same input until $G$ is fully revealed.

  We now analyze the performance of $\textsc{Alg}$ on $G$. Let $x$ denote the
  total number of vertex deletions algorithm $\textsc{Alg}$ performs on $G$ with
  advice $u_t(G)$. We already showed that $x \ge m \cdot k$. If $G$ is not
  expanded to $G'$, the advice $u_t(G)$ is potentially incorrect. Furthermore,
  we know that $\text{cost}(\textsc{Opt}(G)) = m$. Since $m$ can be chosen
  arbitrarily large, it follows that
  \[ w_{\textsc{Alg}} \ge \frac{\text{cost}(\textsc{Alg}^{u_t(G)}(G))}{\text{cost}(\textsc{Opt}(G))} = \frac{x}{m} \ge k\,. \]
  
  $\textsc{Alg}$ deletes during its execution on $G$ with advice
  $u_t(G)$ for each of the $m$ presented gadgets at least $k-1$ vertices with
  advice $0$, as shown above, and therefore in total at least $m(k-1)$ vertices
  with advice $0$. We denote the number of deleted vertices with advice $1$ by
  $i$, which is bounded by the total number of vertex deletions $x$ minus the
  number of deleted vertices with advice $0$. Therefore, $i \le x - m(k-1)$.
  Furthermore, we showed above that there is at any point of the execution of
  algorithm $\textsc{Alg}$ on $G$ with advice $u_t(G)$ at most one undeleted
  vertex with advice $1$ for each gadget. Hence, there are at most $m$
  undeleted vertices with advice $1$ in $G$ after the execution of
  $\textsc{Alg}$ on $G$. We denote the number of undeleted vertices with advice
  $1$ in $G$ after the execution of $\textsc{Alg}$ as $a \le m$. In the
  construction of $G'$ no additional vertices with advice $1$ are introduced as
  shown in \cref{lem:tr-untr}. The total number of vertices with advice $1$ in
  $G'$ with $u_t(G')$ is therefore $i+a$. Since the advice $u_t(G')$ is correct
  by construction of $G'$, the vertices with advice $1$ decode an optimal
  solution. Thus, $\text{cost}(\textsc{Opt}(G')) = i+a$. In \cref{lem:tr-untr},
  $G'$ is constructed by introducing disjoint subgraphs $H$ for each vertex
  with advice $1$. In particular, those do not consist of any other vertex of
  the original graph $G$. Therefore, algorithm $\textsc{Alg}$ has to delete at
  least one additional vertex for each of the $a$ undeleted vertices with
  advice $1$ to get $G'$ $H$-free. Algorithm $\textsc{Alg}$ thus deletes a
  total number of at least $x+a$ vertices during its execution on $G'$ with
  $u_t(G')$. Hence,
  \[ \frac {\text{cost}(\textsc{Alg}^{u_t(G')}(G'))}{\text{cost}(\textsc{Opt}(G'))} \ge \frac{x+a}{i+a} \ge \frac{x+a}{x - m(k-1) +a}\,. \]
  We know that $k>1$, $m w_{\textsc{Alg}} \ge x \ge mk$, $m >0$
  and $m \ge a \ge 0$. It follows that
  \[ m w_{\textsc{Alg}} + m \ge x + a\, , \]
  which implies    
  \[ \frac{m (k-1)}{x-m (k-1)+a} + 1 \ge \frac{k-1}{w_{\textsc{Alg}}-k+2} +1 \]
  and therefore
  \[ \frac{x+a}{x- m(k-1)+a}  \ge \frac{w_{\textsc{Alg}} + 1}{w_{\textsc{Alg}}-k+2}  \,.\]

  Since advice $u_t(G')$ is the only correct advice on $G'$ and $m$ can be
  chosen arbitrarily large, it follows that
  \[ r_{\textsc{Alg}} \ge \frac{\text{cost}(\textsc{Alg}^{u_t(G')} (G'))}{\text{cost}(\textsc{Opt}(G'))} \ge \frac{w_{\textsc{Alg}} + 1}{w_{\textsc{Alg}}-k+2} \,. \]
  Since $w_{\textsc{Alg}} \ge k$ we can substitute $w_{\textsc{Alg}} = p/(1-p)+k$ with $p \in [0,1)$. We get
  \[
    r_{\textsc{Alg}} \ge \frac{w_{\textsc{Alg}} + 1}{w_{\textsc{Alg}}-k+2} = \frac{\frac{p}{1-p}+k + 1}{\frac{p}{1-p}+k-k+2} = \frac{k(2-p) - k +1}{2-p} = k  - \frac{1}{2-p} \cdot (k-1) \,,
  \]
  which concludes the proof.
\end{proof}

Note that the lower bound proven in \cref{thm:boundLower} does not match the
upper bound given by the family of algorithms $\textsc{Alg}_p$
(\cref{alg:Ap}). We later investigate subgraphs $H$ for which a better
algorithm is possible to exist. But for subgraphs $H$ where we can combine
the used gadgets such that the number of leftover vertices with advice $1$ is
constant, we are able to prove that the family of algorithms $\textsc{Alg}_p$
is Pareto-optimal. This is the case for those subgraphs $H$ which are
additionally at least $2$-vertex-connected.

\begin{theorem} \label{thm:boundLowerTight}
  Let $H$ be a connected subgraph that does not contain two true twins or
  does not contain two false twins. If $H$ is $2$-vertex-connected, then
  there does not exist any deterministic algorithm solving the
  \textsc{Delayed Connected $H$-Node-Deletion Problem} with predictions
  that has both a consistency of less than $k-p \cdot (k-1)$ and a robustness
  of less than $k + p/(1-p)$ for any $p \in [0,1)$ and $k = |H| > 1$.
\end{theorem}

\begin{proof}
  We will expand the idea used in the proof of \cref{thm:boundLower} by
  constructing an instance that keeps the number of undeleted vertices with
  advice $1$ constant. An adversary constructs the following online graph $G$
  with advice $u_t(G)$ for an arbitrary algorithm $\textsc{Alg}$ that solves
  the \textsc{Delayed Connected $H$-Node-Deletion Problem} for a forbidden
  subgraph $H$ and is not arbitrarily bad. Recall the gadgets $g^i$ from the
  proof of \cref{lem:unadv_cl} for subgraphs $H$ without two false twins or of
  \cref{lem:unadv_tr} for subgraphs $H$ without two true twins. The adversary
  designs advice for each vertex $v \in V^i$ of such a gadget again by
  assigning advice $1$ to vertex $v$ if $v \in V^i_1$ and advice $0$ otherwise.
  Recall that $V^i_j$ is the equivalence class such that $v \in V^i_j$ if
  $v=v_j$ or $v$ was reinserted for some $v'$ and $v' \in V^i_j$. We already
  showed that the properties of \cref{rem:properties} hold for any algorithm
  with such advice in the proof of \cref{thm:boundLower}. It follows that there
  is exactly one vertex with advice $1$ in each induced copy of $H$ in $g^i$
  and there is at most one undeleted vertex with advice $1$ in $g^i$.
  
  The adversary constructs $G$ by choosing $m>0$ such gadgets
  $g^{i_1},\dots,g^{i_m}$ and combining them in the following way. The first
  gadget $g^{i_1}$ is presented normally. If there is no undeleted vertex with
  advice $1$ after the execution of algorithm $\textsc{Alg}$ on gadget
  $g^{i_j}$ for $j \in [1,m-1]$, the next gadget $g^{i_{j+1}}$ is presented
  disjoint from the rest of the graph. Otherwise, there is one undeleted vertex
  $v$ with advice $1$ in gadget $g^{i_j}$ after the execution of the algorithm
  on it. The next gadget $g^{i_{j+1}}$ will now use vertex $v$ as the first
  vertex of the presented copy of $H$. Note that this vertex would receive
  advice $1$ by the construction of the advice anyways. If $v$ is deleted
  during the construction of $g^{i_{j+1}}$, the vertex which is reinserted for
  $v$ only shares edges with the other vertices in $g^{i_{j+1}}$ and does not
  copy the edges between $v$ and any previous vertices of gadget $g^{i_j}$. All
  other presented vertices in $g^{i_{j+1}}$ are disjoint to the rest of the
  graph. This ensures that there is always at most one undeleted vertex with
  advice $1$ during the execution of $\textsc{Alg}$ on $G$ because every time
  an undeleted vertex with advice $1$ remains after the execution on one
  gadget, it is reused for the next gadget. 
  
  It is clear that each gadget, even when combined in this manner, is still
  able to enforce that at least all $k$ vertices of the copy of $H$, that is
  originally presented, have to be deleted by algorithm $\textsc{Alg}$, if
  there are no induced copies of $H$ between different gadgets, which we prove
  later.  Thus, $\textsc{Alg}$ deletes a total number of $x \ge m \cdot k$
  vertices on $G$ with advice $u_t(G)$. At least $m(k-1)$ of those deleted
  vertices have advice $0$ and thus, at most $x-m(k-1)$ of those have advice
  $1$. We now show that, with the given properties of subgraph $H$, only $m$
  vertex deletions are required to make the graph $G$ $H$-free, i.e.,
  $\text{cost}(\textsc{Opt}(G))=m$. It is easy to see that
  $\text{cost}(\textsc{Opt}(G)) \ge m$. For this consider the copy of $H$ that
  is originally presented in each of the $m$ gadgets. Because the gadget
  enforces that all of its $k$ vertices are deleted, any leftover undeleted
  vertex, which is used in the next gadget, cannot be part of it. Since all
  gadgets are otherwise disjoint, those $m$ copies of $H$ are also disjoint and
  hence require at least $m$ vertex deletions. Now, suppose that
  $\text{cost}(\textsc{Opt}(G)) > m$ to show a contradiction. We know that, by
  construction of each gadget from \cref{lem:unadv_cl} (resp.
  \cref{lem:unadv_tr}), one vertex deletion suffices to delete all induced
  copies of $H$ in this gadget. If $\text{cost}(\textsc{Opt}(G)) > m$, then
  there must be a remaining copy of $H$ in graph $G$ after the deletion of
  those in total $m$ vertices, which are optimal for each gadget. Thus, this
  remaining copy of $H$ must include vertices from at least two different
  gadgets, which are not part of both gadgets. We know that subgraph $H$ must
  be connected. Because the gadgets are only connected through at most one
  vertex $v$ (the reused vertex with advice $1$), this remaining copy of $H$
  must include vertex $v$ and at least one other vertex from each of the two
  different gadgets, since $v$ is part of both gadgets. But since $v$ is the
  only connection between those two gadgets, deleting $v$ would lead to those
  other vertices, which are part of the copy of $H$, being unconnected. This is
  a direct contradiction to subgraph $H$ being 2-vertex-connected.  Hence, it
  follows that $\text{cost}(\textsc{Opt}(G)) = m$ and because $m$ can be chosen
  arbitrarily large, $w_{\textsc{Alg}} \ge \frac{x}{m} \ge k$.
  
  The remaining part of the proof works similar to the previous proof of
  \cref{thm:boundLower}. Because there are no copies of $H$ spanning over
  multiple gadgets, there is a vertex with advice $1$ in each copy of $H$ in
  the graph $G$ and we can expand the online graph $G$ to an online graph $G'$
  such that $u_t(G')$ is a unique correct advice according to
  \cref{lem:tr-untr}. Algorithm $\textsc{Alg}$ cannot detect if $G$ is expanded
  to $G'$ or not and will therefore perform the same until all vertices of
  graph $G$ are fully revealed. The total number of by $\textsc{Alg}$ deleted
  vertices with advice $1$, which we denote by $i$, on $G$ with advice $u_t(G)$
  is again bounded by $i \le x-m(k-1)$ as shown above. We also showed that the
  number of undeleted vertices with advice $1$ is at most one. If $G$ is
  expanded to $G'$, no additional vertices with advice $1$ are added and, since
  the advice $u_t(G')$ is correct, the vertices with advice $1$ encode an
  optimal solution. Hence, $\text{cost}(\textsc{Opt}(G')) \le i+1$. We again
  know that by the construction of $G'$ in \cref{lem:tr-untr}, algorithm
  $\textsc{Alg}$ has to delete at least one additional vertex to get $G'$
  $H$-free because of the copies of $H$ that are newly introduced for the
  undeleted vertex with advice $1$ in $G'$. Thus, $\text{cost}(
  \textsc{Alg}^{u_t(G')}(G') \ge x+1$. And hence,
  \[ \frac{\text{cost}(\textsc{Alg}^{u_t(G')} (G'))}{\text{cost}(\textsc{Opt}(G'))} \ge \frac{x+1}{i+1} \ge \frac{x+1}{x-m(k-1)+1}\,. \]
  We know that $k>1$, $mw_{\textsc{Alg}} \ge x \ge mk$ and $m>0$. It follows that 
  $m w_{\textsc{Alg}} + 1 \ge x + 1$, which implies
  \[ \frac{m (k-1)}{x-m (k-1)+1} \ge \frac{m (k-1)}{m w_{\textsc{Alg}}-m(k-1)+1} = \frac{k-1}{w_{\textsc{Alg}}-k+1+\frac{1}{m}} \]
  and thus
  \[ \frac{x+1}{x- m(k-1)+1} \ge \frac{w_{\textsc{Alg}}+\frac{1}{m}}{w_{\textsc{Alg}}-k+1+\frac{1}{m}}  \ge  \frac{w_{\textsc{Alg}}}{w_{\textsc{Alg}}-k+1+\frac{1}{m}} \,.\]

  Since advice $u_t(G')$ is the only correct advice on $G'$ and $m$ can be
  chosen arbitrarily large, it follows that
  \[ r_{\textsc{Alg}} \ge \frac{\text{cost}(\textsc{Alg}^{u_t(G')} (G'))}{\text{cost}(\textsc{Opt}(G'))} \ge \frac{w_{\textsc{Alg}}}{w_{\textsc{Alg}}-k+1+\frac{1}{m}}\ge \frac{w_{\textsc{Alg}}}{w_{\textsc{Alg}}-k+1}\,. \]
  We can now substitute $w_{\textsc{Alg}} = p/(1-p)+k$ with $p \in [0,1)$ since
  $w_{\textsc{Alg}} \ge k$, yielding
  \[
      r_{\textsc{Alg}} \ge \frac{w_{\textsc{Alg}}}{w_{\textsc{Alg}}-k+1} = \frac{\frac{p}{1-p}+k}{\frac{p}{1-p}+k-k+1} = \frac{p+k (1-p)}{p + 1-p} = k-p \cdot (k-1)\,.
  \]
\end{proof}
  
Let us now take a closer look at a simple family of graphs that do not satisfy
the conditions of \cref{thm:boundLowerTight}, since they are not $2$-vertex
connected: the paths $P_k$ of constant length $k$. First, we investigate why
the graph $G$ constructed in the proof of \cref{thm:boundLowerTight} fails to
enforce the matching bounds for $H=P_3$.

Consider the graph $G$ built in \cref{thm:boundLowerTight} for $H=P_3$. It is
clear that each gadget is still able to enforce $k$ vertex deletions, but when
combining the gadgets to $G$, they lose the property that they only require one
optimal vertex deletion per gadget. This happens because new induced copies of
$P_3$ appear between the combined gadgets as is shown in \cref{fig:paths}.
Hence, combining $m$ such gadgets leads to $\text{cost}(\textsc{Opt}(G))>m$
which is incompatible with the key part of the proof. Nevertheless, this does
not happen for $H=P_2$ and we are able to change the way we combine the gadgets
such that we can also show a matching lower bound for $H = P_k$ for $k \ge 5$.
Therefore, the family of algorithms $\textsc{Alg}_p$ (\cref{alg:Ap}) proves to be
Pareto-optimal for those forbidden subgraphs, too. 

It is easy to see that the proof of \cref{thm:boundLowerTight} directly applies
for $H=P_2$ because induced copies of a forbidden subgraph $H$ can only appear
between the gadgets if $H$ consists of at least three vertices. This is
because it would need at least one vertex in the old gadget, one in the new
gadget, and the shared vertex.
  
\begin{figure}[t]
  \begin{center}
      \definecolor{yqyqyq}{rgb}{0.5019607843137255,0.5019607843137255,0.5019607843137255}
      \begin{tikzpicture}[line cap=round,line join=round,x=0.7cm,y=0.7cm]
        \def\radius{1.5pt};
        \draw[line width=1.pt,dash pattern=on 8pt off 8pt,color=yqyqyq] (-4.,4.)-- (-4.,0.);
        \draw[line width=1.pt,dash pattern=on 8pt off 8pt,color=yqyqyq] (-4.,0.)-- (0.,0.);
        \draw[line width=1.pt,dash pattern=on 8pt off 8pt,color=yqyqyq] (0.,0.)-- (4.,0.);
        \draw[line width=1.pt,dash pattern=on 8pt off 8pt,color=yqyqyq] (0.,4.)-- (0.,0.);
        \draw[line width=1.pt,dash pattern=on 8pt off 8pt,color=yqyqyq] (0.,4.)-- (-4.,0.);
        \draw[line width=1.pt,dash pattern=on 8pt off 8pt,color=yqyqyq] (0.,4.)-- (4.,0.);
        \draw[line width=1.pt,dash pattern=on 8pt off 8pt,color=yqyqyq] (-4.,4.)-- (0.,0.);
        \draw[line width=1.pt] (-4.,4.)-- (0.,4.);
        \draw[line width=1.pt] (-4.,4.)-- (-8.,4.);
        \draw[line width=1.pt] (-8.,4.)-- (-12.,4.);
        \node[color=yqyqyq, fill, label={[text=yqyqyq]below:{$v_1^1,1$}}, circle, inner sep=\radius] at (-4,0) {};
        \node[color=yqyqyq, fill, label={[text=yqyqyq]below:{$v_2^1,0$}}, circle, inner sep=\radius] at (0,0) {};
        \node[color=yqyqyq, fill, label={[text=yqyqyq]below:{$v_3^1,0$}}, circle, inner sep=\radius] at (4,0) {};
        \node[fill, label=above:{${v'}^1_2,0$}, circle, inner sep=\radius] at (0.,4.) {};
         \node[fill, label=above:{${v'}^{1,2}_1,1$}, circle, inner sep=\radius] at (-4.,4.) {};
        \node[fill, label=above:{$v^2_2,0$}, circle, inner sep=\radius] at  (-8.,4.) {};
        \node[fill, label=above:{$v^2_3,0$}, circle, inner sep=\radius] at (-12.,4.) {};
    \end{tikzpicture}
  \end{center}
  \caption{The combination of two gadgets $g^{i_1}$ and $g^{i_2}$ during a possible execution of an algorithm according to \cref{thm:boundLowerTight} for $H=P_3$. The advice of each vertex is given after its name and the corresponding gadget is given in the superscript of the name. Vertex ${v'}^{1,2}_1$ is part of both gadgets. Grey vertices indicate that they were deleted by the algorithm and the incident edges of deleted vertices are displayed as dashed. There are clearly induced paths of length three between both gadgets, e.g.\ $v^2_2, {v'}^{1,2}_1, v^1_1$.}
  \label{fig:paths}
\end{figure}
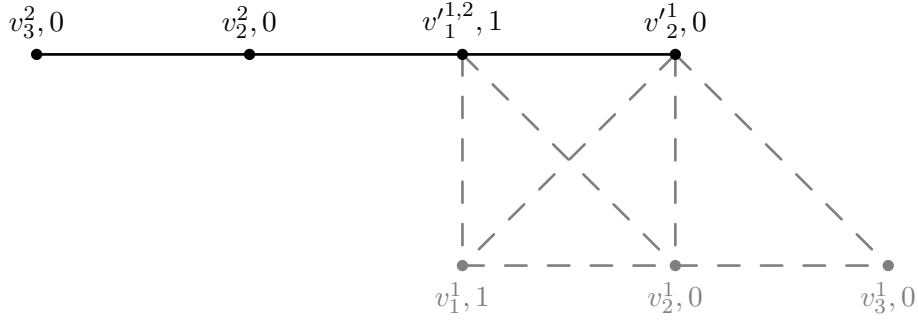

\begin{lemma}\label{lemma:paths}
  For $H = P_k$ and $k \ge 5$, there does not exist any deterministic algorithm
  solving the $\textsc{Delayed $H$-Node-Deletion Problem}$ that has both a
  consistency of less than $k-p \cdot (k-1)$ and a robustness of less than
  $k+p/(1-p)$ for any $p\in [0,1)$.
\end{lemma}

\begin{proof}
  The key part of the proof is to find a way to combine the gadgets $g^i$
  (see \cref{lem:unadv_tr}) such that $\text{cost}(\textsc{Opt}(G))=m$
  without losing the desired properties of each gadget and while keeping
  the number of undeleted vertices with advice $1$ constant. Again, the
  adversary is able to do this by preventing the occurrence of induced
  copies of $P_k$ between the gadgets while still reusing the undeleted
  vertices with advice $1$ for $k \ge 5$. The adversary constructs advice
  for each vertex $v$ of a gadget by assigning advice $1$ to $v$ if $v \in
    V^i_1$ and advice $0$ otherwise, as seen before.
    
  The adversary combines $m > 0$ such gadgets $g^{i_1},\dots,g^{i_m}$ to form
  graph $G$ as follows. If there is no undeleted vertex with advice $1$ after
  the execution of algorithm $\textsc{Alg}$ on gadget $g^{i_j}$ for $1 \le j
  <m$, then the next gadget $g^{i_{j+1}}$ is presented vertex-disjoint from
  the rest of the graph. Otherwise, there is an undeleted vertex $v$ with
  advice $1$ in the gadget $g^{i_j}$ after the execution of the algorithm
  $\textsc{Alg}$ on gadget $g^{i_j}$ and gadget $g^{i_{j+1}}$ is introduced
  in the following way. Gadget  $g^{i_{j+1}}$ reuses vertex $v$ with advice
  $1$ as its first vertex of the presented copy of $P_k$. The gadgets are
  tweaked such that every vertex of gadget $g^{i_{j}}$ shares an edge with
  every vertex which is not part of gadget $g^{i_{j}}$. Note that the reused
  vertex $v$ is both part of $g^{i_j}$ and $g^{i_{j+1}}$ and that for any two
  different gadgets, there is at most one reused vertex which is part of both
  gadgets. The edges between different gadgets do not change the property
  that any online algorithm $\textsc{Alg}$, that is not arbitrarily bad, has
  to delete at least $k$ vertices, $k-1$ of those with advice $0$, during its
  execution on each gadget, if there are no induced copies of $P_k$ between
  the gadgets. We now prove this. Suppose there is an induced copy of $P_k$
  in $G$ that includes at least two vertices that are part of different
  gadgets and not part of the same gadget. Let us denote those vertices by
  vertex $v^x_1$, which is part of the gadget $x$, and vertex $v^y_1$, which
  is part of the gadget $y$. Furthermore, we denote the set of $k \ge 5$
  vertices that form this induced copy of $P_k$ by $V_P$.  Because $v^x_1 \in
  V_P$ and $v^y_1 \in V_P$ are not part of the same gadget, there is an edge
  between $v^x_1$ and $v^y_1$. There cannot be a vertex in $V_P$ which is not
  part of gadget $x$ or gadget $y$ because it would share an edge with
  $v^x_1$ and $v^y_1$ and hence form a triangle with those, which cannot be
  part of any induced $P_k$. Since there is at most one vertex that is part
  of both gadgets $x$ and $y$, there must be two other vertices in $V_P$
  which are part of gadgets $x$ and $y$ but not part of both. Suppose those
  two vertices are part of the same gadget. Without loss of generality, let
  them be part of gadget $x$. Therefore, they both share an edge with $v^y_1$
  which leads to vertex $v^y_1$ having a degree of at least three in the
  subgraph induced by $V_P$, which is not possible since that subgraph should
  be a path of length $k$. Thus, there must be a vertex $v^y_2 \in V_P$,
  which is part of gadget $y$ and not part of gadget $x$, and a vertex $v^x_2
  \in V_P$, which is part of gadget $x$ and not part of gadget $y$. Thus,
  there is a cycle of length $4$ in the subgraph induced by $V_P$, namely
  $v^x_1, v^y_1, v^x_2, v^y_2, v^x_1$ which is again a direct contradiction.
  Thus, there is no induced copy of $P_k$ between two different gadgets in
  $G$ and deleting the optimal vertex of each of the $m$ gadgets makes $G$
  $P_k$-free. Hence, $\text{cost}(\textsc{Opt}(G))=m$. The rest of the proof
  works analogously to the proof of \cref{thm:boundLowerTight}. 
\end{proof}

In the proof of \Cref{lemma:paths} we were able to combine the gadgets to graph $G$ in such a
way that no induced copies of $H$ occur between different gadgets. Thus,
$\text{cost}(\textsc{Opt}(G))=m$. There also might be another way to combine
gadgets to a graph $G$ with $\text{cost}(\textsc{Opt}(G))=m$, where there are
induced copies of $H$ between different gadgets, but each of those copies
includes at least one optimal vertex of a gadget. Nevertheless, this is not
possible for $H=P_3$ which we can show as following under some assumptions.

Consider a gadget $g^i$ with a leftover, undeleted vertex $v$ with advice $1$
after the execution of some algorithm. The adversary now wants to use this
vertex $v$ for another gadget by introducing a new $P_3$ which includes
vertex $v$, to keep the number of undeleted vertices with advice $1$
constant. It seems unavoidable for the gadget $g^i$ that $v$ is adjacent to
at least one vertex $w$ of gadget $g^i$ that is not optimal for $g^i$ but
will be deleted by some algorithm. This is the case because if a certain
algorithm deletes the vertex with advice $1$ first for a freshly presented
copy of $P_3$, to create a new intact copy, the vertex $v$ with advice $1$
has to be introduced such that it is adjacent to at least another vertex $w$.
The algorithm now chooses to not delete this vertex $v$ and to delete the
other vertex $w$ instead. Since the algorithm has deleted only two vertices,
vertex $w$ should not be optimal for the gadget because, if it were, the
gadget could not force another vertex deletion without introducing a second
optimal vertex. This behavior of an algorithm does not seem exploitable with
the gadgets we previously introduced. We now show that if we create a new
gadget that includes $v$ by introducing a new induced path $P_3$ with
vertices $v$, $v_1$ and $v_2$, there is always another induced path between a
vertex of this new path and vertex $w$, thus creating an induced path between
different gadgets. Note that $w$ cannot be part of the new gadget because it
is already deleted. Let us first consider the case where $v$ is the middle
vertex of the induced path $v_1,v,v_2$  and thus shares an edge with $v_1$
and $v_2$. If there is no edge between $v_1$ and $w$ or $v_2$ and $w$,
vertices $w,v,v_1$ or $w,v,v_2$ form an induced path of length three. If
there is an edge $v_1,w$ and $v_2,w$, vertices $v_1,w,v_2$ form an induced
path of length three. If vertex $v$ is an outer vertex of the path, we fix
$v_1$ to be the middle vertex. We already showed that there always is an
induced path between the gadgets if there is an edge between $v_1$ (the
middle vertex) and $w$. If there is no edge $(v_1,w)$, vertices  $w,v,v_1$
form an induced path of length three. We just showed that the adversary could
not avoid the occurrence of an induced path between different gadgets. It is
also easy to see that this path does not necessarily include an optimal
vertex of one gadget. Vertex $w$ is not optimal and while one out of $v_1$,
$v_2$ or $v$ will be optimal for the next gadget, an algorithm can choose to
delete the vertices included in the path with $w$ first and thus force those
vertices to be not optimal. This would force the combination of those two
gadgets to have an optimal solution greater than two.

Therefore, the subgraph $H= P_3$ is another candidate (along with the one we
saw in \cref{fig:windmill}) for which the family of algorithms $\textsc{Alg}_p$
(\cref{alg:Ap}) may not be optimal.

\section{Conclusion}

We presented a family of algorithms that turns out to be Pareto-optimal for
many subgraphs $H$ on the $\textsc{Delayed $H$-Node-Deletion Problem}$ under the
model of algorithms with predictions. This family of algorithms can only be Pareto-optimal
on subgraphs for which the na\"{\i}ve $k$-competitive algorithm is optimal for the
problem without advice. Thus, we also showed this property for many forbidden
subgraphs $H$ by using suitable adversarially constructed gadgets $g^i$.
Additionally, we proposed a subgraph $H$ on which those gadgets are not able to
enforce a competitive ratio of at least $k$ for every algorithm without advice,
indicating that a better algorithm might exist both with and without untrusted
advice. Furthermore, we showed that there might be a better family of
algorithms with advice for the forbidden subgraph $P_3$, even though the na\"{\i}ve
algorithm is optimal for it in the model without advice.

Further work is required to investigate those proposed subgraphs and to find an
even more general rule determining whether the given algorithms are optimal for
a forbidden subgraph $H$. Future research might also examine if these results
also hold for other forms of predictions, for example for advice of any
form and length. One might also consider the problem of non-induced forbidden
subgraphs.

\section*{Acknowledgments}
The authors would like to thank Fabian Frei, Matthias Gehnen, and Peter Rossmanith for inspiring this paper by proposing the algorithm $\textsc{Alg}_p$ for the special case $H=P_2$, corresponding to the \textsc{Delayed Online Vertex Cover Problem}.

\end{document}